\documentclass[12pt,reqno]{article}

\usepackage[usenames]{color}
\usepackage{amssymb}
\usepackage{amsmath}
\usepackage{amsthm}
\usepackage{amsfonts}
\usepackage{amscd}
\usepackage{graphicx}

\usepackage[colorlinks=true,
linkcolor=webgreen,
filecolor=webbrown,
citecolor=webgreen]{hyperref}

\definecolor{webgreen}{rgb}{0,.5,0}
\definecolor{webbrown}{rgb}{.6,0,0}

\usepackage{color}
\usepackage{fullpage}
\usepackage{float}

\usepackage{graphics}
\usepackage{latexsym}
\usepackage{epsf}
\usepackage{breakurl}

\usepackage{slashbox}

\newcommand{\seqnum}[1]{\href{https://oeis.org/#1}{\rm \underline{#1}}}

\begin{document}

\theoremstyle{plain}
\newtheorem{theorem}{Theorem}
\newtheorem{corollary}[theorem]{Corollary}
\newtheorem{lemma}[theorem]{Lemma}
\newtheorem{proposition}[theorem]{Proposition}

\theoremstyle{definition}
\newtheorem{definition}[theorem]{Definition}
\newtheorem{example}[theorem]{Example}
\newtheorem{conjecture}[theorem]{Conjecture}

\theoremstyle{remark}
\newtheorem{remark}[theorem]{Remark}

\title{Hilbert's spacefilling curve described by
automatic, regular, and synchronized sequences}

\author{Jeffrey Shallit\\
School of Computer Science\\
University of Waterloo\\
Waterloo, ON N2L 3G1 \\
Canada\\
\href{mailto:shallit@uwaterloo.ca}{\tt shallit@uwaterloo.ca}}

\maketitle

\begin{abstract}
We describe Hilbert's spacefilling curve in several different ways:
as an automatic sequence of directions,
as a regular and synchronized sequence of coordinates of 
lattice points encountered, and as an automatic bitmap image.
\end{abstract}

\section{Introduction}

In 1891 David Hilbert famously described the construction of  a continuous
curve that fills the unit square \cite{Hilbert:1891}.   So many
papers on this topic have been published since then
(for example, see \cite{Butz:1969,Butz:1971,Sagan:1994,Breinholt&Schierz:1998})
that it seems difficult to say anything new
about it.  Nevertheless, we'll try.   We will describe the curve
in three different ways:   as a $4$-automatic sequence, as a
$4$-regular sequence, and as a $(4,2,2)$-synchronized sequence.
An interesting feature of our approach is that in each case,
we ``guess'' the correct representation, and then use the theorem-prover
{\tt Walnut} to prove our guess is correct \cite{Mousavi:2016}.

Instead of filling the unit square, we will treat a version that
visits every non-negative pair of integers, starting from
the origin $(0,0)$.   At each stage we take the figure
constructed so far, make four copies, flip each copy
appropriately (Figure~\ref{fig1}), and join them together, as illustrated in 
Figure~\ref{fig2}.

Let us agree to write $\tt U$ for up, $\tt D$ for down, 
$\tt R$ for right, and
$\tt L$ for left.  Thus $A_n$, the $n$'th generation of the curve, 
can be written as a string over the alphabet $\{ {\tt U,D,R,L}\}$,
and it is easy to see that in fact $|A_n|= 4^n - 1$.
Note that the moves inserted to connect the pieces
depend on the parity:   to go from
$A_n$ to $A_{n+1}$ when $n$ is odd, we successively
insert $\tt RUL$ to connect the pieces,
but when $n$ is even, we successively insert $\tt URD$.

The first four generations of the curve are encoded as follows:
\begin{align*}
A_0 &= \epsilon \\
A_1 &= {\tt URD} \\
A_2 &= {\tt URDRRULURULLDLU} \\
A_3 &= {\tt URDRRULURULLDLUURULUURDRURDDLDRRRULUURDRURDDLDRDDLULLDRDLDRRURD}
\end{align*}

Notice that $A_n$ is a prefix of $A_{n+1}$ for
all $n \geq 0$.
So we can let 
$${\bf HC} = (h_n)_{n \geq 0} = {\tt URDRRULURULLD} \cdots $$
be the unique infinite string of which $A_1, A_2, A_3, \ldots$ are
all prefixes.

Furthermore $A_n$ is a path from $(0,0)$ to $(2^n-1,0)$ if $n$
is odd, and a path from $(0,0)$ to $(0,2^n - 1)$ if $n$ is
even.  
\begin{figure}[htb]
\begin{center}
\includegraphics[width=6in]{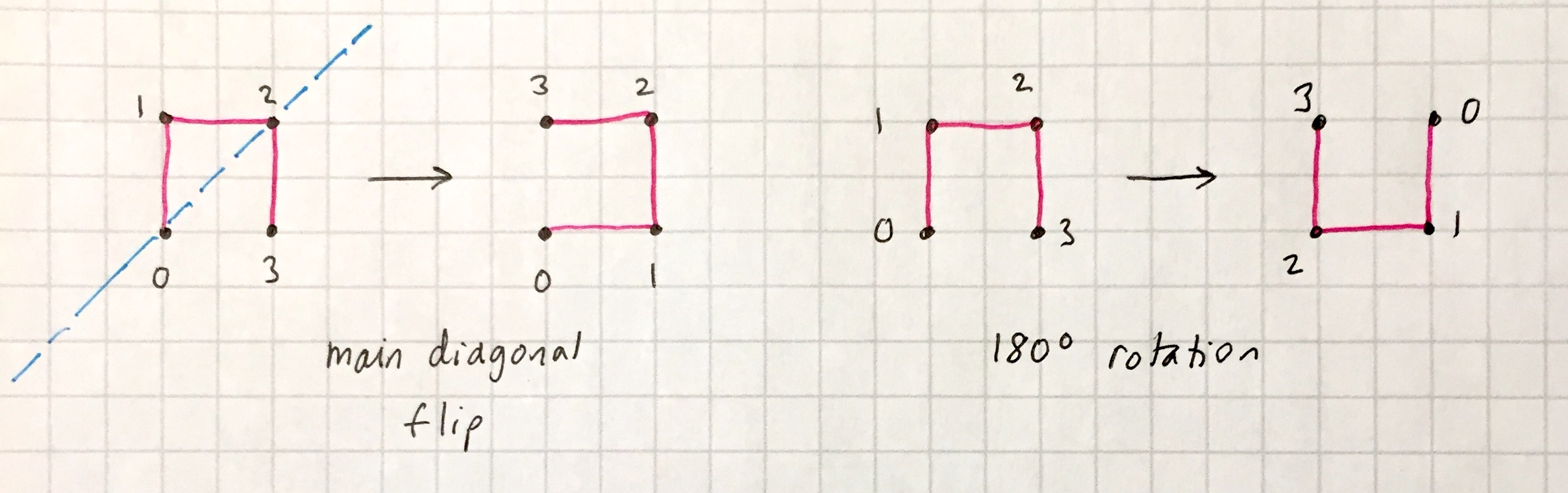}
\end{center}
\caption{How one generation follows from the previous.}
\label{fig1}
\end{figure}

\section{By recurrence}

We now wish to write a recurrence for the $A_n$.  We define 
the codings $t_D , t_H $ as follows:   $t_D$ is a flip about
the main diagonal, and hence $t_D({\tt UDRL}) = \tt RLUD$, while
$t_H$ is a $180^\circ$ rotation, that is,
$t_H ({\tt UDRL}) = \tt DULR$.   

This gives us a formula to compute $A_n$, namely
\begin{align}
A_0 &= \epsilon \nonumber\\
A_{2n+1} & = A_{2n} \, {\tt U} \, t_D(A_{2n}) \, {\tt R} \, t_D(A_{2n}) \, 
{\tt D} \, t_H(A_{2n}) \label{an} \\
A_{2n+2} &= A_{2n+1} \, {\tt R} \, t_D(A_{2n+1}) \, {\tt U} \, t_D(A_{2n+1})  \, {\tt L} \, t_H(A_{2n+1}) \nonumber
\end{align}
for $n \geq 0$.

\begin{figure}[H]
\begin{center}
\includegraphics[width=5.5in]{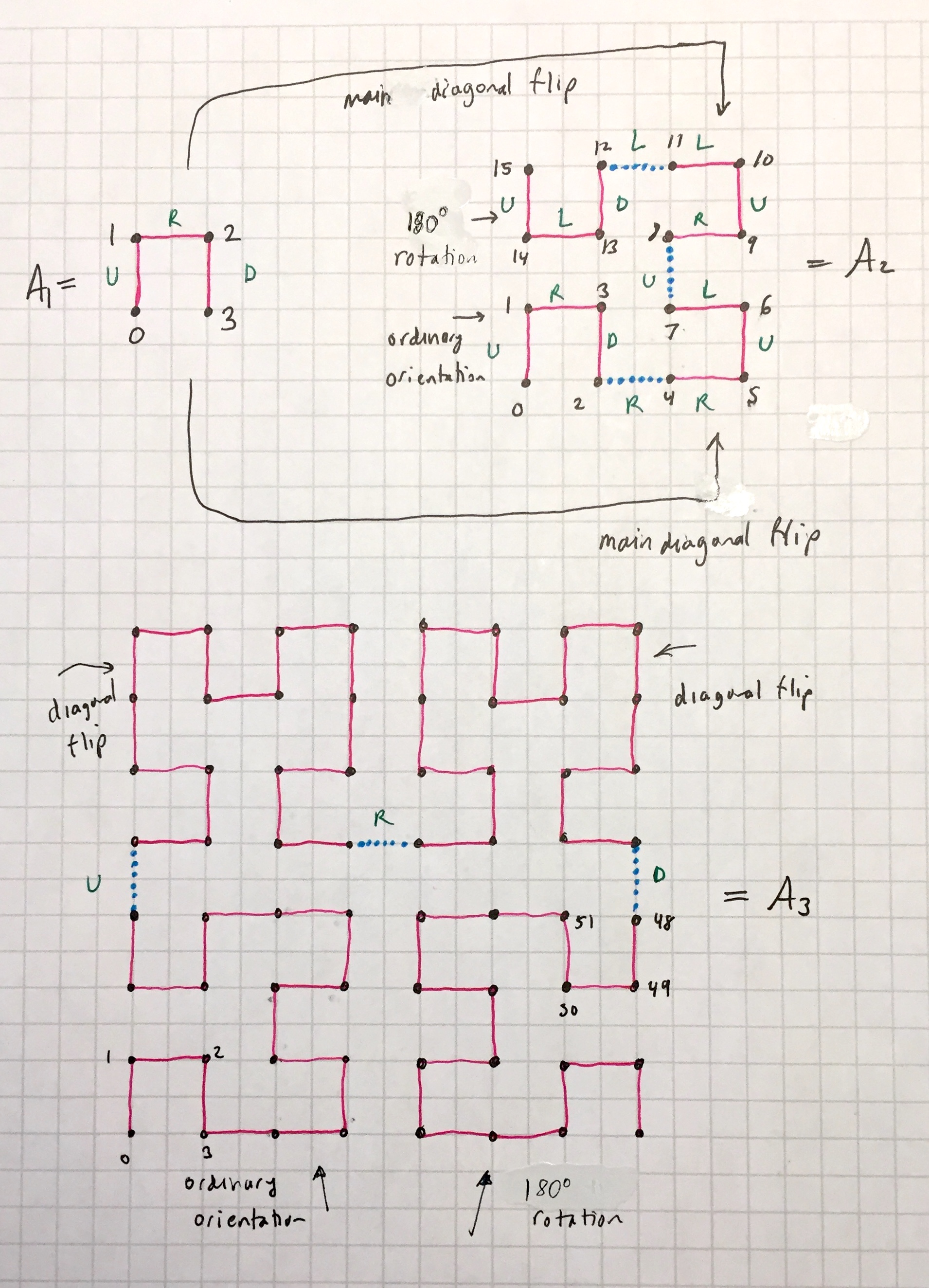}
\end{center}
\caption{The first three generations of the Hilbert curve.}
\label{fig2}
\end{figure}

\section{As image of fixed point of morphism or automatic sequence}
\label{three}

We use the notation $\Sigma_k = \{ 0,1,\ldots, k-1 \}$, and
$(n)_k$ denotes the canonical base-$k$ representation of $n$, starting
with the most significant digit.   

Recall that a sequence $(a_n)_{n \geq 0}$ is $k$-automatic if its
{\it $k$-kernel\/}
$$ \{ (a_{k^e n + i})_{n \geq 0} \, :\, e \geq 0, 0 \leq i < k^e \}$$
is of finite cardinality \cite[\S 6.6]{Allouche&Shallit:2003}.
Alternatively, 
$(a_n)_{n \geq 0}$ is $k$-automatic if
there is a deterministic
finite automaton with output (DFAO) 
that reads $(n)_k$ as input and reaches a state with output $a_n$.

The fact that the length of the
$i$'th generation $A_i$ is so close to $4^i$ strongly suggests
that $H$ might be $4$-automatic.   

To try to determine the DFAO, we can use a ``guessing procedure'' 
based on the Myhill-Nerode theorem \cite[\S 3.9]{Shallit:2009} to find
a good candidate, and then use a theorem-prover to prove that
our guess is correct.   We repeat this strategy throughout the paper.

We find an $8$-state DFAO
$(Q, \Sigma, \Gamma, \delta, q_0, \tau)$ as follows:
\begin{itemize}
\item $Q = \{ 0, 1, \ldots, 7 \}$;
\item $\Sigma = \Sigma_4$;
\item $\Gamma = \{ {\tt U, R, D, L} \}$;
\item $q_0 = 0$;
\item $\delta(q,i)$ and $\tau$ are defined as in Table~\ref{tab4}.
\begin{table}[H]
\begin{center}
\begin{tabular}{c|cccc|c}
\backslashbox{$q$}{$i$} & 0 & 1 & 2 & 3 & $\tau(q)$ \\
\hline
0 & 0 & 1 & 2 & 3 & {\tt U} \\
1 & 1 & 0 & 4 & 5 & {\tt R} \\
2 & 1 & 0 & 4 & 6 & {\tt D} \\
3 & 7 & 6 & 5 & 0 & {\tt R}\\
4 & 0 & 1 & 2 & 7 & {\tt L}\\
5 & 6 & 7 & 3 & 1 & {\tt U} \\
6 & 6 & 7 & 3 & 2 & {\tt L} \\
7 & 7 & 6 & 5 & 4 & {\tt D}
\end{tabular}
\end{center}
\caption{DFAO for the sequence $\bf HC$.}
\label{tab4}
\end{table}
\end{itemize}

In Walnut this DFAO can be represented by the name
{\tt HC}.  Because {\tt Walnut} currently does not allow letters
as output, we use the recoding of the output given by the correspondence
${\tt U} \leftrightarrow {\tt 0}$,
${\tt R} \leftrightarrow {\tt 1}$,
${\tt D} \leftrightarrow {\tt 2}$,
${\tt L} \leftrightarrow {\tt 3}$.

We can verify that this automaton is correct by using {\tt Walnut}.
From Eq.~\eqref{an}, it suffices to check that
for all $n \geq 0$ we have
\begin{align}
{\bf HC}[0] &= {\tt U} \label{eq1} \\
{\bf HC}[4^n..2\cdot 4^n - 1] &= t_D ( {\bf HC}[0..4^n-1] ) \label{eq2} \\
{\bf HC}[2\cdot 4^n..3\cdot 4^n - 2] &= t_D({\bf HC}[0..4^n-2])  \label{eq3} \\
{\bf HC}[3 \cdot 4^n - 1] &= \begin{cases} 
	{\tt L}, & \text{if $n$ odd;} \\
	{\tt D}, & \text{if $n$ even;} 
	\end{cases} \label{eq4} \\
{\bf HC}[3 \cdot 4^n..4^{n+1} - 2] &= t_H({\bf HC}[0..4^n-2]) \label{eq5} \\
{\bf HC}[4^n - 1] &= \begin{cases}
	{\tt R}, & \text{if $n$ odd;} \\
	{\tt U}, & \text{if $n$ even.}
	\end{cases} \label{eq6}
\end{align}
which we can do with {\tt Walnut} as follows:
\begin{verbatim}
reg power4 msd_4 "0*10*":
reg evenpower4 msd_4 "0*1(00)*":
reg oddpower4 msd_4 "0*10(00)*":
eval test2 "?msd_4 HC[0]=@0":
eval test3 "?msd_4 Ax,t ($power4(x) & t<x) =>
     ((HC[t]=@0 <=> HC[x+t]=@1) & (HC[t]=@1 <=> HC[x+t]=@0)
     &(HC[t]=@2 <=> HC[x+t]=@3) & (HC[t]=@3 <=> HC[x+t]=@2))":
eval test4 "?msd_4 Ax,t ($power4(x) & t+1<x) =>
     ((HC[t]=@0 <=> HC[2*x+t]=@1) & (HC[t]=@1 <=> HC[2*x+t]=@0)
     &(HC[t]=@2 <=> HC[2*x+t]=@3) & (HC[t]=@3 <=> HC[2*x+t]=@2))":
eval test5 "?msd_4 Ax ($oddpower4(x) => HC[3*x-1]=@3) &
     ($evenpower4(x) => HC[3*x-1]=@2)":
eval test6 "?msd_4 Ax,t ($power4(x) & t+1<x) =>
     ((HC[t]=@0 <=> HC[3*x+t]=@2) & (HC[t]=@1 <=> HC[3*x+t]=@3)
     &(HC[t]=@2 <=> HC[3*x+t]=@0) & (HC[t]=@3 <=> HC[3*x+t]=@1))":
eval test7 "?msd_4 Ax ($oddpower4(x) => HC[x-1]=@1) &
     ($evenpower4(x) => HC[x-1]=@0)":
\end{verbatim}
\noindent and everything returns {\tt true}.

\section{As system of coordinates and a $4$-regular sequence}

Recall that a $k$-regular sequence is a generalization of
automatic sequence.   Being $k$-regular
means there is a finite subset $S$ of the $k$-kernel
such that each element of the $k$-kernel can
be written as a linear combination of elements of $S$
\cite{Allouche&Shallit:1992,Allouche&Shallit:2003b}.

Suppose we start at $(x_0,y_0) := (0,0)$ and perform unit steps
according to the letters specified by $\bf HC$.
Thus we have
$$
(x_{n+1}, y_{n+1}) = (x_n,y_n) + 
	\begin{cases} 
	(1,0), & \text{ if ${\bf HC}[n] = {\tt R}$}; \\
	(-1,0), & \text{ if ${\bf HC}[n] = {\tt L}$}; \\
	(0,1), & \text{ if ${\bf HC}[n] = {\tt U}$}; \\
	(0,-1), & \text{ if ${\bf HC}[n] = {\tt D}$}; 
	\end{cases}
$$

This
gives us a sequence of ordered pairs specifying the $x$-$y$ coordinates
of the $n$'th point along the curve.    
Table~\ref{tab3} gives the
few values of $((x_n, y_n))_{n \geq 0}$.   The
sequence $(x_n)_{n \geq 0}$ is sequence \seqnum{A059252} 
and $(y_n)_{n \geq 0}$ is \seqnum{A059253} in \cite{Sloane:2021}.
\begin{table}[H]
\begin{center}
\begin{tabular}{c|cccccccccccccccc}
$n$ & 0 & 1 & 2 & 3 & 4 & 5 & 6 & 7 & 8 & 9 & 10 & 11 & 12 & 13 & 14 & 15\\
\hline
$x_n$ & 0 & 0 & 1 & 1 & 2 & 3 & 3 & 2 & 2 & 3 & 3 & 2 & 1 & 1 & 0 & 0\\
$y_n$ & 0 & 1 & 1 & 0 & 0 & 0 & 1 & 1 & 2 & 2 & 3 & 3 & 3 & 2 & 2 & 3\\
\end{tabular}
\end{center}
\caption{First few values of $(x_n,y_n)$ for the Hilbert curve.}
\label{tab3}
\end{table}
In this section we show that $(x_n,y_n)_{n \geq 0}$,
the sequence of coordinates traversed
by the Hilbert curve, is $4$-regular.  Actually, this follows 
immediately from \cite[Theorem 3.1]{Allouche&Shallit:1992}, but
applying this theorem is somewhat messy.

Recall that a linear representation for a $k$-regular
sequence $(a_n)_{n \geq 0}$ consists of a row vector $v$,
a matrix-valued morphism $\gamma$, and a column vector $w$
such that $a_n = v \gamma( (n)_k ) w$ for all $n \geq 0$.
The dimension of $w$ is called the {\it rank\/} of
the linear representation; see \cite{Berstel&Reutenauer:2010}.

A ``guessing procedure'' for $k$-regular sequences 
suggests that the $4$-kernel of $(x_n)_{n\geq 0}$ is 
contained in the linear span 
of the $5$ subsequences
$$ \{ (x_{n})_{n\geq 0},\ (x_{4n})_{n \geq 0},\ (x_{4n+1})_{n \geq 0},\ 
(x_{4n+2})_{n \geq 0},\ (x_{16n})_{n \geq 0} \},$$
and the same for $(y_n)_{n \geq 0}$.   
%(Basically, we consider each
%sequence in the $k$-kernel in turn, and try to guess relations with
%previous elements of the $k$-kernel.  If no relation can be deduced,
%we decimate by replacing $n$ with $kn+a$ for $0 \leq a < k$ and continue.)
We can then ``guess'' a number of candidate relations for elements 
of the $4$-kernel for both $x_n$ and $y_n$.
Assuming the guessed relations are correct, by standard techniques,
we can deduce a rank-$5$ base-$4$ linear representation
for $((x_n, y_n))_{n \geq 0}$, namely:
$$ (x_n, y_n) = 
v \gamma( (n)_4 ) w,$$
where $(n)_4$ is the base-$4$ representation of $n$ and
$$v =
\left[
\begin{array}{ccccc}
0&0&0&1&0\\
0&0&1&1&0
\end{array}
\right];
\quad
\gamma(0) =
\left[
\begin{array}{ccccc}
0& 0& 0& 0&-4\\
 1& 0&-1&-1& 4\\
0& 0& 1& 0& 0 \\
0& 0& 0& 1& 0 \\
0& 1& 1& 1& 1 \\
 \end{array}
 \right]; \quad
\gamma(1) =
\left[
\begin{array}{ccccc}
0& 0& 0& 0&-4\\
 0& 0& 0&-1& 0\\
 1&-2&-3&-2& 4\\
0& 2& 3& 3& 0 \\
0& 1& 1& 1& 1
\end{array}
\right];
$$
$$
\gamma(2) =
\left[
\begin{array}{ccccc}
 0& 0& 0& 0&-4\\
 0&-2&-2&-3& 0\\
 0& 0&-1& 0& 0\\
1& 2& 3& 3& 4\\
0& 1& 1& 1& 1
\end{array}
\right]; \quad
\gamma(3) =
\left[
\begin{array}{ccccc}
 0& 0& 0& 0&-4\\
 1&-3&-2&-2& 1\\
-1& 2& 1& 2&-4\\
1& 1& 1& 0& 7 \\
0& 1& 1& 1& 1
\end{array}
\right];
\quad
w =
\left[
\begin{array}{ccccc}
1\\
0\\
0\\
0\\
0\\
\end{array}
\right].$$

We now prove that this is indeed a linear representation
for $((x_n, y_n))_{n \geq 0}$, in a somewhat roundabout way. 

Recall that a finite-state transducer $T$ maps input strings to output
strings.   The output associated with an input are the string or
strings arising
from the concatenation of the outputs of all transitions, provided
processing the string ends in a final state of $T$.   We allow
a transducer to be nondeterministic.  A transducer
is {\it functional\/} if
every input results in at most one output \cite[\S 3.5]{Shallit:2009}.
We need a lemma.
\begin{lemma}
Let $(f(n))_{n \geq 0}$ be a $k$-regular sequence, and let
$\Sigma_k = \{ 0,1,\ldots, k-1 \}$.
Let $T = (Q, \Sigma_k, \Sigma_k, \delta, q_0, F, \rho)$ be a
nondeterministic functional finite-state transducer with transitions on
single letters only, but allowing arbitrary words as outputs on
each transition.  More precisely,
\begin{itemize}
\item $Q = \{ q_0, \ldots, q_{r-1} \}$;
\item $\delta:Q \times \Sigma_k \rightarrow Q$ is the
transition function; and
\item $\rho:Q \times \Sigma_k \rightarrow \Sigma_k^*$ is
the output function;
\item $F$ is a set of final states.
\end{itemize}
Let the domain of $\delta$ and $\rho$ be extended to $\Sigma_k^*$ in the obvious way.
Define $g(n) = f(T((n)_k)))$.
Then $(g(n))_{n \geq 0}$ is also a $k$-regular sequence.
\end{lemma}

\begin{proof}
Let $(v, \mu, w)$ be a rank-$s$ linear representation for $f$.
We create a linear representation $(v', \mu', w')$ for $g$.

The idea is that $\mu' (a)$, $0 \leq a < k$,
is an $n \times n$ matrix, where
$n = rs$.  It is easiest to think of $\mu'(a)$ as an $r \times r$
matrix, where each entry is itself an $s \times s$ matrix.
In this interpretation, $(\mu'(a))_{i,j} = \mu(\rho(q_i, a)) $
if $\delta(q_i, a) = q_j$.

An easy induction now shows that
if $\delta(q_i,x) = q_j$ and $\rho(q_i, x) = y$, then
$(\mu'(x))_{i,j} = \mu(y)$.
If we now let $v'$ be the vector $[v \quad 0 \quad \cdots \quad 0]$
and $w'$ be the column vector with $w$'s in the positions of the final
states and $0$'s otherwise,
then it follows that
$v' \mu'(x) w' = v \mu(T(x)) w$.
This gives a linear representation for $(g(n))_{n \geq 0}$.
\end{proof}

In particular, there is a simple finite-state transducer that, when
applied to the base-$k$ representation of $n$ gives the representation of 
$n+1$.  
For $k =4$ this transducer is depicted in Figure~\ref{trans} below:
\begin{figure}[H]
\begin{center}
\includegraphics[width=5in]{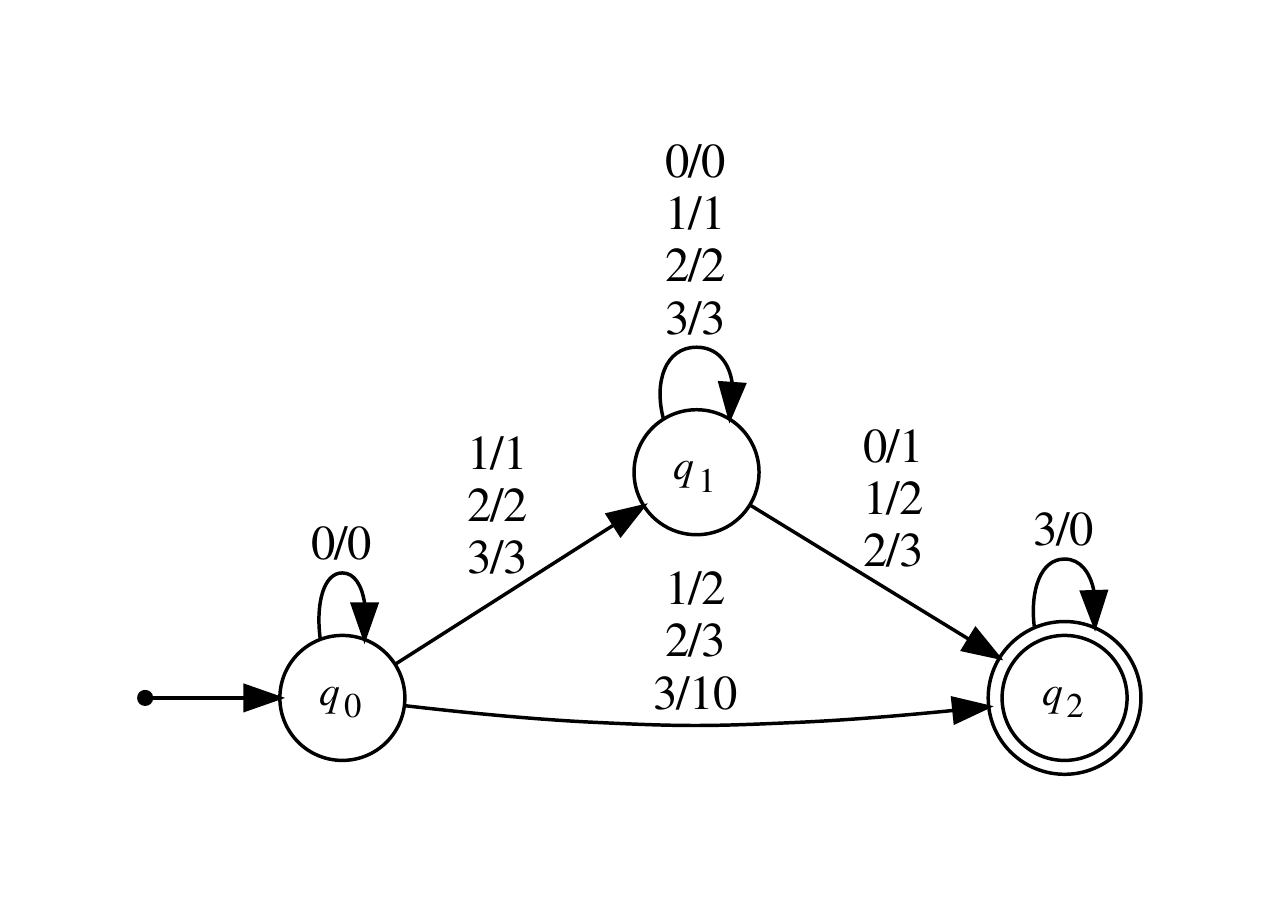}
\end{center}
\caption{A transducer mapping $(n)_4$ to $(n+1)_4$}
\label{trans}
\end{figure}
So from the guessed linear representation for
$(x_n)_{n \geq 0}$ we easily deduce a
linear representation for $(x_{n+1})_{n \geq 0}$.  It is of 
rank $15$.

Now from these two linear representations, we can use 
an obvious ``tensor product''-style construction to get the linear
representation for the first difference sequence $(x_{n+1} -x_n)_{n \geq 0}$.
This produces a linear representation of rank $20$.
We can do the same thing with $(y_n)_{n \geq 0}$.

Now we can use the Berstel-Reutenauer minimization algorithm 
\cite[\S 3.3]{Berstel&Reutenauer:2010} for
linear representations to get an equivalent
minimized linear representation $(v', \gamma', w')$ for
$((x_{n+1} - x_n, y_{n+1}-y_n))_{n \geq 0}$.  Here is what it looks like:
$$v' =
\left[
\begin{array}{ccc}
1 & 0 & 0 \\
0 & 1 & 0
\end{array}
\right];
\quad
\gamma'(0) =
\left[
\begin{array}{ccc}
1 & 0 & 0 \\
0 & 1 & 0 \\
0 & 1 & 0 
\end{array}
 \right]; \quad
\gamma'(1) =
\left[
\begin{array}{ccccc}
0 & 1 & 0 \\
1 & 0 & 0 \\
1 & 0 & 0
\end{array}
\right];
$$
$$
\gamma'(2) =
\left[
\begin{array}{ccc}
0 & 0 & 1\\
1 & -1 & 1 \\
1 & -1 & 1 
\end{array}
\right]; \quad
\gamma'(3) =
\left[
\begin{array}{ccc}
-1 & 1 & -1 \\
0 & 0 & -1 \\
0 & -1 & 0
\end{array}
\right]; 
\quad
w' =
\left[
\begin{array}{ccc}
0\\
1\\
0\\
\end{array}
\right].$$
Now we use the ``semigroup trick"
(see \cite[pp.~951,954]{Du&Mousavi&Schaeffer&Shallit:2016})
to find an automaton for the
first difference sequence $((x_{n+1} - x_n,y_{n+1}-y_n))_{n \geq 0}$ and prove
that the resulting automaton has only finitely many states.
And no surprise---it is the
same automaton we started with, the one in Section~\ref{three}.
This shows that our 
guessed linear representation for $((x_n, y_n))_{n \geq 0}$
was indeed correct.

The advantage to the representation as a $4$-regular sequence is that
we can compute $(x_n, y_n)$ in time linear in the number of bits of
$n$:   we express $n$ in base 4, and then multiply the appropriate
vectors and $O(\log_4 n)$ matrices.

\section{As a synchronized function}

Finally, perhaps the most interesting representation of
the Hilbert curve is that 
$n$ and $(x_n, y_n)$ are synchronized, but only if we
represent $n$, $x_n$, and $y_n$ in the right way.    The right way
is to represent $n$ in base $4$, but represent $x_n$ and $y_n$
in base $2$!  In other words, the triple $(n, x_n, y_n)$ is
$(4,2,2)$-synchronized \cite{Carpi&Maggi:2001}.

Here our guessing procedure 
guesses a 10-state automaton ${\tt HS} =
(Q, \Sigma, \delta, q_0, F)$, given below.
Here
\begin{itemize}
\item $Q = \{0,1,\ldots, 9 \}$;
\item $\Sigma = \Sigma_4 \times \Sigma_2 \times \Sigma_2$;
\item $q_0 = 0$;
\item $F = \{ 0, 2, 3, 5, 6, 7\}$;
\end{itemize}
and $\delta$ is represented in Table~\ref{syncha}.  (All transitions
not listed go to a dead state that is not
accepting, which just loops to itself on each input.)

\begin{itemize}
\item verify that indeed $\tt HS$ represents a synchronized
function: 
	\begin{itemize}
	\item for each $n$ there is a pair $(x,y)$ such that
${\tt HS}[n][x][y]$ is true
	\item for each $n$ there is only one pair $(x,y)$
such that ${\tt HS}[n][x][y]$ is true
	\end{itemize}
\item verify that ${\tt HS}[0][0][0]=0$;
\item verify that if ${\tt HS}[n][x][y]$ and ${\tt HS}[n+1][x'][y']$
both hold, then $(x'-x, y'-x)$ corresponds to the 
appropriate move ${\tt U, D, R, L}$ computed by the
automaton in Section~\ref{three}.
\end{itemize}

{\footnotesize
\begin{verbatim}
eval fn1 "An Ex,y HS[?msd_4 n][x][y]=@1":
# f(n) takes an ordered pair value for each n

eval fn2 "An,x,y,xp,yp (HS[?msd_4 n][x][y]=@1 & HS[?msd_4 n][xp][yp]=@1) => (x=xp & y=yp)":
# f(n) takes only one value for each n

eval check_up "An (HC[?msd_4 n]=@0 <=> Ex,xp,y,yp HS[?msd_4 n][x][y]=@1 &
    HS[?msd_4 n+1][xp][yp]=@1 & xp=x & yp=y+1)":

eval check_right "An (HC[?msd_4 n]=@1 <=> Ex,xp,y,yp HS[?msd_4 n][x][y]=@1 &
    HS[?msd_4 n+1][xp][yp]=@1 & xp=x+1 & yp=y)":

eval check_down "An (HC[?msd_4 n]=@2 <=> Ex,xp,y,yp HS[?msd_4 n][x][y]=@1 &
    HS[?msd_4 n+1][xp][yp]=@1 & xp=x & yp+1=y)":

eval check_left "An (HC[?msd_4 n]=@3 <=> Ex,xp,y,yp HS[?msd_4 n][x][y]=@1 &
    HS[?msd_4 n+1][xp][yp]=@1 & xp+1=x & yp=y)":
\end{verbatim}
}
\noindent and {\tt Walnut} returns {\tt true}.

Finally, we can use {\tt Walnut} to verify that every pair of natural
numbers $(x,y)$ is hit by one and exactly one
$n$, so our curve is indeed space-filling:
\begin{verbatim}
eval allhit "Ax,y En HS[?msd_4 n][x][y]=@1":
eval hitonce "An,np,x,y (HS[?msd_4 n][x][y]=@1 & HS[?msd_4 np][x][y]=@1)
     => (?msd_4 n=?msd_4 np)":
\end{verbatim}
\noindent and {\tt Walnut} returns {\tt true} for both.

\begin{table}[H]
\begin{center}
\begin{tabular}{c|c|c||c|c|c}
$q$ & $t=[i,j,k]$ & $\delta(q,t)$ & $q$ & $t=[i,j,k]$ & $\delta(q,t)$ \\
\hline
0 & $[0,0,0]$ & 0 & 5 & $[0,0,0]$ & 5 \\
0 & $[1,0,1]$ & 3 & 5 & $[1,0,1]$ & 9 \\
0 & $[1,1,0]$ & 1 & 5 & $[1,1,0]$ & 6 \\
0 & $[2,1,1]$ & 5 & 5 & $[2,1,1]$ & 0 \\
0 & $[3,0,1]$ & 4 & 5 & $[3,0,1]$ & 7 \\
0 & $[3,1,0]$ & 2 & 5 & $[3,1,0]$ & 8 \\
\hline
1 & $[0,0,0]$ & 3 & 6 & $[0,0,0]$ & 9 \\
1 & $[1,1,0]$ & 6 & 6 & $[1,1,0]$ & 1 \\
1 & $[2,1,1]$ & 6 & 6 & $[2,1,1]$ & 1 \\
1 & $[3,0,1]$ & 7 & 6 & $[3,0,1]$ & 4 \\
\hline
2 & $[0,1,1]$ & 4 & 7 & $[0,1,1]$ & 8 \\
2 & $[1,0,1]$ & 8 & 7 & $[1,1,0]$ & 4 \\
2 & $[2,0,0]$ & 8 & 7 & $[2,0,0]$ & 4 \\
2 & $[3,1,0]$ & 9 & 7 & $[3,0,1]$ & 1 \\
\hline
3 & $[0,0,0]$ & 1 & 8 & $[0,1,1]$ & 7 \\
3 & $[1,0,1]$ & 9 & 8 & $[1,0,1]$ & 2 \\
3 & $[2,1,1]$ & 9 & 8 & $[2,0,0]$ & 2 \\
3 & $[3,1,0]$ & 8 & 8 & $[3,1,0]$ & 3 \\
\hline
4 & $[0,1,1]$ & 2 & 9 & $[0,0,0]$ & 6 \\
4 & $[1,1,0]$ & 7 & 9 & $[1,0,1]$ & 3 \\
4 & $[2,0,0]$ & 7 & 9 & $[2,1,1]$ & 3 \\
4 & $[3,0,1]$ & 6 & 9 & $[3,1,0]$ & 2 \\
\end{tabular}
\end{center}
\caption{The synchronized automaton for coordinates of Hilbert's curve.}
\label{syncha}
\end{table}

From the synchronized automaton, given 
$(n)_4 = a_1 a_2 \cdots a_t$, the base-$4$ representation of $n$,
we can easily determine $(x_n, y_n)$, by intersecting
the automaton with an automaton accepting those
strings of the form $[a_1,*,*][a_2,*,*]\cdots[a_t,*,*]$,
where $*$ denotes either $0$ or $1$.  In the resulting automaton,
only one path is accepting, and it can easily be found in
$O(t)$ time through breadth-first or depth-first search.

But the reverse is also true:  given the base-$2$ representations
of $(x,y)$, we can easily determine the $n$ for which
$(x_n, y_n) = (x,y)$, using the same idea.

\section{As an automatic bitmap image}

With the aid of the synchronized representation for {\tt HS}, we
can easily produce a bitmap image of each generation of the Hilbert
curve, as previously done in \cite[Fig.~6]{Shallit&Stolfi:1989}.

To do so, we ``expand'' the curve, inserting rows and column that
are blank, except for when they connect two consecutive points of the
curve.   The following {\tt Walnut} code produces a DFA
{\tt \$hp} describing a bitmap image of the Hilbert curve.
\begin{verbatim}
def even "Em n=2*m":
def odd "Em n=2*m+1":
def hp "($even(x) & $even(y)) |
($even(x) & $odd(y) &
(En (HS[?msd_4 n][x/2][(y-1)/2]=@1 & HS[?msd_4 n+1][x/2][(y+1)/2]=@1)
|(HS[?msd_4 n][x/2][(y+1)/2]=@1 & HS[?msd_4 n+1][x/2][(y-1)/2]=@1)) |
($odd(x) & $even(y) &
(En (HS[?msd_4 n][(x-1)/2][y/2]=@1 & HS[?msd_4 n+1][(x+1)/2][y/2]=@1)
|(HS[?msd_4 n][(x+1)/2][y/2]=@1 & HS[?msd_4 n+1][(x-1)/2][y/2]=@1))":
\end{verbatim}

For example, for generation $7$ we get the image in Figure~\ref{hilb7}.
\begin{figure}[H]
\begin{center}
\includegraphics[width=3.1in]{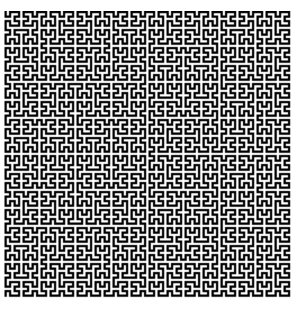}
\end{center}
\caption{Generation $7$ of the Hilbert curve.}
\label{hilb7}
\end{figure}

\section*{Acknowledgments}

Three other descriptions of the Hilbert curve have some
commonalities with the approach given here.
Bially gave a state diagram similar to an automaton
\cite{Bially:1969}.
Gosper \cite[Item 115, pp.~52--53]{Beeler&Gosper&Schroeppel:1972}
gave an iterative measure to determine $(x_n, y_n)$ from
the base-$2$ expansion of $n$.    And Arndt \cite[\S 1.31.1]{Arndt:2011}
also gave a description in terms of iterated morphisms, but 
not quite the same as given here.

I thank Jean-Paul Allouche for helpful discussions.

\end{document}